\documentclass[a4paper,notitlepage]{article}
\usepackage{amssymb}

\newtheorem{theorem}{Theorem}

\newtheorem{corollary}[theorem]{Corollary}

\newtheorem{proposition}[theorem]{Proposition}

\newenvironment{proof}[1][Proof]{\noindent\textbf{#1.} }{\ \rule{0.5em}{0.5em}}
\begin{document}

\title{Autonomous three dimensional Newtonian systems which admit Lie and
Noether point symmetries}
\author{Michael Tsamparlis\thanks{%
Email: mtsampa@phys.uoa.gr}, Andronikos Paliathanasis\thanks{%
Email: anpaliat@phys.uoa.gr}, Leonidas Karpathopoulos\thanks{%
Email: lkarpathopoulos@gmail.com} \\
{\small \textit{Faculty of Physics, Department of Astronomy - Astrophysics -
Mechanics,}}\\
{\small \textit{\ University of Athens, Panepistemiopolis, Athens 157 83,
GREECE}}}
\maketitle

\begin{abstract}
We determine the autonomous three dimensional \ Newtonian systems which
admit Lie point symmetries and the three dimensional autonomous Newtonian
Hamiltonian systems, which admit Noether point symmetries. We apply the
results in order to determine the two dimensional Hamiltonian dynamical
systems which move in a space of constant non-vanishing curvature and are
integrable via Noether point symmetries. The derivation of the results is
geometric and can be extended naturally to higher dimensions.
\end{abstract}


PACS - numbers: 45.20.D-, 02.20.Sv, 02.40.Dr

\section{Introduction}

The Lie and Noether point symmetries of the equations of motion of a
dynamical system provide a systematic method for the determination of
invariants and first integrals (see \cite{Kaushal 1998} for a review). In a
recent work \cite{Tsam10} we have determined the autonomous two dimensional
Newtonian systems which admit Lie and Noether point symmetries. In the
present work we extend this study to the autonomous three dimensional
Newtonian systems. That is, we consider the equations of motion%
\begin{equation}
\ddot{x}^{\mu }=F^{\mu }\left( x^{\nu }\right) ~,~\mu =1,2,3  \label{L2p.1}
\end{equation}%
and compute the form of the functions \textit{\ }$F^{\mu }\left( x^{\nu
}\right) $\ for which (\ref{L2p.1}) admits Lie point symmetries\textit{\ }%
(in addition to the trivial one $\partial _{t}$)\textit{. }

Subsequently we assume the system to be Hamiltonian with Lagrangian
\begin{equation}
L\left( x^{\mu },\dot{x}^{\nu }\right) =\frac{1}{2}\delta _{\mu \nu }\dot{x}%
^{\mu }\dot{x}^{\nu }-V\left( x^{\mu }\right)
\end{equation}%
where $\delta _{\mu \nu }$ is the Euclidian 3d metric and $V\left( x^{\mu
}\right) $ is the potential function and determine the potential functions $%
V\left( x^{\mu }\right) $ for which the Lagrangian admits at least one Lie
or Noether point symmetry (in addition to the trivial $\partial _{t}$).
Because the derivation is based solely on geometric arguments the results
can be generalized in a straightforward manner in $E^{n}.$

Using the fact that a space of constant curvature of dimension $n-1$ can be
embedded in a flat space of dimension $n,$ we apply the results in $E^{3}$
in order to determine the dynamical systems which move in a two dimensional
space of constant non-vanishing curvature and are Liouville integrable via
Noether point symmetries.

The structure of the paper is as follows. In section \ref{Collineations of
Riemannian spaces} we give the basic definitions concerning the
collineations in a Riemannian space. In section \ref{Lie and Noether point
symmetries versus Collineations} we present two theorems which relate the
Lie and the Noether point symmetry algebras of the equations of motion of a
dynamical system moving in an $n-$ dimensional Riemannian space with the
projective and the homothetic algebra of the space respectively. In section %
\ref{Lie point symmetries of three dimensional autonomous Newtonian systems}
we determine the autonomous Newtonian systems which admit Lie point
symmetries. In section \ref{Noether point symmetries1} we determine the
subset of the systems which admit Noether point symmetries. In section \ref%
{Motion on the two dimensional sphere} we apply the results to determine the
Newtonian Hamiltonian dynamical systems which move in a two dimensional
space of constant non-vanishing curvature and admit Noether point
symmetries. Finally in section\ \ref{Conclusion} we draw our conclusions.

\section{Collineations of Riemannian spaces}

\label{Collineations of Riemannian spaces}

A\ collineation in a Riemannian space is a vector field $\mathbf{X}$ which
satisfies an equation of the form%
\begin{equation}
\mathcal{L}_{X}\mathbf{A=B}  \label{L2p.2}
\end{equation}%
where $\mathcal{L}_{X}$ denotes Lie derivative \cite{Yano}, $\mathbf{A}$ is
a geometric object (not necessarily a tensor)\ defined in terms of the
metric and its derivatives (e.g. connection coefficients, Ricci tensor,
curvature tensor etc.) and $\mathbf{B}$ is an arbitrary tensor with the same
tensor indices as $\mathbf{A}$. The collineations in a Riemannian space have
been classified by Katzin et al. \cite{Katzin}. In the following we use only
certain collineations.

A conformal Killing vector (CKV)\ is defined by the relation
\begin{equation}
\mathcal{L}_{X}g_{ij}=2\psi \left( x^{k}\right) g_{ij}.
\end{equation}%
If $\psi =0,$ $\mathbf{X}\ $is called a Killing vector (KV), if $\psi $ is a
non-vanishing constant $\mathbf{X}$ is a homothetic vector (HV) and if $\psi
_{;ij}=0,~\mathbf{X}$~is called a special conformal Killing vector (SCKV). A
CKV is called proper if it is not a KV, HV or a SCKV.

A Projective collineation (PC) is defined by the equation
\begin{equation}
\mathcal{L}_{X}\Gamma _{jk}^{i}=2\phi _{(,j}\delta _{k)}^{i}.
\end{equation}%
If $\phi =0$ the PC\ is called an affine collineation (AC) and if $\phi
_{;ij}=0$ a special projective collineation (SPC). A\ proper PC\ is a PC\
which is not an AC, HV or KV or SPC. The PCs form a Lie algebra whose ACs,
HVs and KVs are subalgebras. It has been shown that if a metric admits a
SCKV then also admits a SPC, a gradient HV and a gradient KV \cite{HallR}.

In the following we shall need the symmetry algebra of spaces of constant
curvature. In \cite{Barnes} it has been shown that the PCs of a space of
constant non-vanishing curvature consist of proper PCs and KVs only and if
the space is flat then the algebra of the PCs consists of KVs/HV/ACs and
SPCs. Note that the algebra of KVs is common in both cases.

\section{Lie and Noether point symmetries versus Collineations}

\label{Lie and Noether point symmetries versus Collineations}

We review briefly the basic definitions concerning Lie and Noether point
symmetries of systems of second order ordinary differential equations (ODEs)%
\begin{equation}
\ddot{x}^{i}=\omega ^{i}\left( t,x^{j},\dot{x}^{j}\right) .  \label{Lie.0}
\end{equation}

A\ vector field $X=\xi \left( t,x^{j}\right) \partial _{t}+\eta ^{i}\left(
t,x^{j}\right) \partial _{i}$\ in the augmented space $\{t,x^{i}\}$ is the
generator of a Lie point symmetry of the system of ODEs (\ref{Lie.0})\ if
the following condition is satisfied \cite{Olver}
\begin{equation}
X^{\left[ 2\right] }\left( \ddot{x}^{i}-\omega ^{i}\left( t,x^{j},\dot{x}%
^{j}\right) \right) =0  \label{Lie.1}
\end{equation}%
where $X^{\left[ 2\right] }$\ is the second prolongation of $X$ defined by
the formula%
\begin{equation}
X^{\left[ 2\right] }=\xi \partial _{t}+\eta ^{i}\partial _{i}+\left( \dot{%
\eta}^{i}-\dot{x}^{i}\dot{\xi}\right) \partial _{\dot{x}^{i}}+\left( \ddot{%
\eta}^{i}-\dot{x}^{i}\ddot{\xi}-2\ddot{x}^{i}\dot{\xi}\right) \partial _{%
\ddot{x}^{i}}.  \label{Lie.2}
\end{equation}%
Condition (\ref{Lie.1}) is equivalent to the condition \cite{StephaniB}
\begin{equation}
\left[ X^{\left[ 1\right] },A\right] =\lambda \left( x^{a}\right) A
\label{Lie.3}
\end{equation}%
where $X^{\left[ 1\right] }$ is the first prolongation of $X$ and $A$ is the
Hamiltonian vector field%
\begin{equation}
A=\partial _{t}+\dot{x}\partial _{x}+\omega ^{i}\left( t,x^{j},\dot{x}%
^{j}\right) \partial _{\dot{x}^{i}}.  \label{Lie.4}
\end{equation}

If the system of ODEs results from a first order Lagrangian $L=L\left(
t,x^{j},\dot{x}^{j}\right) ,$\ then a Lie symmetry $X$\ is a Noether
symmetry of the Lagrangian, if the additional condition is satisfied
\begin{equation}
X^{\left[ 1\right] }L+L\frac{d\xi }{dt}=\frac{dG}{dt}  \label{Lie.5}
\end{equation}%
where $G=G\left( t,x^{j}\right) $\ is the Noether gauge function. To every
Noether symmetry there corresponds a first integral (a Noether integral)$~$%
\cite{StephaniB} of the system of equations (\ref{Lie.0}) which is given by
the formula%
\begin{equation}
I=\xi E-\frac{\partial L}{\partial \dot{x}^{i}}\eta ^{i}+G  \label{Lie.6}
\end{equation}%
where $E$ \ is the Hamiltonian.

Using the standard Lie method the Lie point symmetry conditions~(\ref{Lie.1}%
) for the second order system%
\begin{equation}
\ddot{x}^{i}+\Gamma _{jk}^{i}\dot{x}^{j}\dot{x}^{k}+F^{i}\left( x^{j}\right)
=0  \label{PP.01}
\end{equation}%
are computed in the following geometric form\footnote{%
The use of an algebraic computing program (e.g. Lie) does not reveal
directly the Lie symmetry conditions in this geometric form. The "solution"\
of these conditions is given in \cite{Tsam10}. For the convenience of the
reader we repeat this solution in concise form.}%
\begin{equation}
L_{\eta }F^{i}+2\xi ,_{t}F^{i}+\eta ^{i},_{tt}=0  \label{PLS.09}
\end{equation}%
\begin{equation}
\left( \xi ,_{k}\delta _{j}^{i}+2\xi ,_{j}\delta _{k}^{i}\right) F^{k}+2\eta
^{i},_{t|j}-\xi ,_{tt}\delta _{j}^{i}=0  \label{PLS.10}
\end{equation}%
\begin{equation}
L_{\eta }\Gamma _{(jk)}^{i}=2\xi ,_{t(j}\delta _{k)}^{i}  \label{PLS.11}
\end{equation}%
\begin{equation}
\xi _{(,i|j}\delta _{r)}^{k}=0.  \label{PLS.12}
\end{equation}

Equation (\ref{PLS.12}) means that $\xi _{,j}$ is a gradient Killing vector
(KV) of $g_{ij}.$ Equation (\ref{PLS.11}) means that $\eta ^{i}$ is a
projective collineation of the metric with projective function $\xi _{,t}.$
The remaining two equations are the constraint conditions, which relate the
components $\xi ,n^{i}$ of the Lie point symmetry vector with the vector $%
F^{i}$. Equation (\ref{PLS.09}) gives%
\begin{equation}
\left( L_{\eta }g^{ij}\right) F_{j}+g^{ij}L_{\eta }F_{j}+2\xi
_{,t}g^{ij}F_{j}+\eta _{,tt}^{i}=0.  \label{de.21a}
\end{equation}%
This equation restricts $\eta ^{i}$ further because it relates it directly
to the metric symmetries. Finally equation (\ref{PLS.10}) gives%
\begin{equation}
-\delta _{j}^{i}\xi _{,tt}+\left( \xi _{,j}\delta _{k}^{i}+2\delta
_{j}^{i}\xi _{,k}\right) F^{k}+2\eta _{,tj}^{i}+2\Gamma _{jk}^{i}\eta
_{,t}^{k}=0.  \label{de.21d}
\end{equation}

We conclude that the Lie symmetry equations are equations (\ref{de.21a}) ,(%
\ref{de.21d}) where $\xi (t,x)$ is a gradient KV of the metric $g_{ij}$ and $%
\eta ^{i}\left( t,x\right) $ is a special projective collineation of the
metric $g_{ij}$ with projective function $\xi _{,t}$. The above lead to the
following Theorem which relates the Lie point symmetries of an autonomous
dynamical system 'moving' in a Riemannian space with the collineations of
the space\footnote{%
This theorem contains various cases which can be found in the detailed
version of the theorem given in \cite{Tsam10}. It is important for the
comprehension of the present paper that the reader will consult the detailed
version of the theorem.}.

\begin{theorem}
\label{The general conservative system} The Lie point symmetries of the
equations of motion of an autonomous dynamical system moving $~$in a
Riemannian space with metric $g_{ij}$ (of any signature) under the action of
the force $F^{i}(x^{j})$ (\ref{PP.01})~are given in terms of the generators $%
Y^{i}$ of the special projective Lie algebra of the metric $g_{ij}$.
\end{theorem}

If the force $F^{i}$ is derivable from a potential $V(x^{i}),$ so that the
equations of motion follow from the standard Lagrangian
\begin{equation}
L\left( x^{j},\dot{x}^{j}\right) =\frac{1}{2}g_{ij}\dot{x}^{i}\dot{x}%
^{j}-V\left( x^{j}\right)  \label{NPC.02}
\end{equation}%
with Hamiltonian%
\begin{equation}
E=\frac{1}{2}g_{ij}\dot{x}^{i}\dot{x}^{j}+V\left( x^{j}\right)
\label{NPC.02a}
\end{equation}%
then the Noether point conditions (\ref{Lie.5}) for the Lagrangian (\ref%
{NPC.02}) are
\begin{eqnarray}
V_{,k}\eta ^{k}+V\xi _{,t} &=&-f_{,t}  \label{NSCS.4} \\
\eta _{,t}^{i}g_{ij}-\xi _{,j}V &=&f_{,j}  \label{NSCS.5} \\
L_{\eta }g_{ij} &=&2\left( \frac{1}{2}\xi _{,t}\right) g_{ij}  \label{NSCS.6}
\\
\xi _{,k} &=&0.  \label{NSCS.7}
\end{eqnarray}

Equation (\ref{NSCS.7}) implies $\xi =\xi \left( t\right) $ and reduces the
system as follows%
\begin{eqnarray}
L_{\eta }g_{ij} &=&2\left( \frac{1}{2}\xi _{,t}\right) g_{ij}  \label{NSCS.8}
\\
V_{,k}\eta ^{k}+V\xi _{,t} &=&-f_{,t}  \label{NSCS.9} \\
\eta _{i,t} &=&f_{,i}.  \label{NSCS.10}
\end{eqnarray}

Equation (\ref{NSCS.8}) implies that $\eta ^{i}$ is a conformal Killing
vector of the metric provided $\xi _{,t}\neq 0.$ Because $g_{ij}$\ is
independent of $t$\ and $\xi =\xi \left( t\right) $\ the $\eta ^{i}$\ must
be is a HV of the metric. This means that $\eta ^{i}\left( t,x\right)
=T\left( t\right) Y^{i}\left( x^{j}\right) $\ where $Y^{i}$\ is a HV. If $%
\xi _{,t}=0$ then $\eta ^{i}$ is a Killing vector of the metric. Equations (%
\ref{NSCS.9}), (\ref{NSCS.10}) are the constraint conditions, which the
Noether symmetry and the potential must satisfy for former to be admitted.
These lead to the following theorem\footnote{%
The detailed version of this theorem is given in \cite{Tsam10}.}

\begin{theorem}
\label{The Noether Theorem}The Noether point symmetries of the Lagrangian (%
\ref{NPC.02}) are generated from the homothetic algebra of the metric $%
g_{ij} $.
\end{theorem}

More specifically, concerning the Noether symmetries, we have the following
\cite{Tsam10}:

All autonomous systems admit the Noether symmetry $\partial _{t}~$whose
Noether integral is the Hamiltonian~$E$. For the rest of the Noether
symmetries we consider the following cases

\textbf{Case I }\ Noether point symmetries generated by the homothetic
algebra.

The Noether symmetry vector and the Noether function $G\left( t,x^{k}\right)
$ are%
\begin{equation}
\mathbf{X}=2\psi _{Y}t\partial _{t}+Y^{i}\partial _{i}~,~G\left(
t,x^{k}\right) =pt  \label{NPC.03}
\end{equation}%
where $\psi _{Y}$ is the homothetic factor of $Y^{i}~$($\psi _{Y}=0$ for a
KV\ and $1$ for the HV) and $p$ is a constant, provided the potential
satisfies the condition%
\begin{equation}
\mathcal{L}_{Y}V+2\psi _{Y}V+p=0.  \label{NPC.04}
\end{equation}

\textbf{Case II} \ Noether point symmetries generated by the gradient
homothetic Lie algebra i.e. both KVs and the HV are gradient. \

In this case the Noether symmetry vector and the Noether function are%
\begin{equation}
\mathbf{X}=2\psi _{Y}\int T\left( t\right) dt\partial _{t}+T\left( t\right)
H^{i}\partial _{i}~~,~G\left( t,x^{k}\right) =T_{,t}H\left( x^{k}\right)
~+p\int Tdt  \label{NPC.05}
\end{equation}%
where $H^{i}$ is a gradient HV or gradient KV, the function $T(t)$ is
computed from the relation~$~T_{,tt}=mT~\ $where $~m$ is a constant and the
potential satisfies the condition
\begin{equation}
\mathcal{L}_{H}V+2\psi _{Y}V+mH+p=0.  \label{NPC.06}
\end{equation}

Concerning the Noether integrals we have the following result (not including
the Hamiltonian)

\begin{corollary}
\label{The Noether Integrals}The Noether integrals (\ref{Lie.6}) of Case I
and Case II are respectively
\begin{equation}
I_{C_{I}}=2\psi _{Y}tE-g_{ij}Y^{i}\dot{x}^{j}+pt  \label{NPC.07}
\end{equation}%
\begin{equation}
I_{C_{II}}=2\psi _{Y}\int T\left( t\right) dt~E-g_{ij}H^{,i}\dot{x}%
^{j}+T_{,t}H+p\int Tdt.  \label{NPC.08}
\end{equation}%
where $E$ is the Hamiltonian (\ref{NPC.02a}).
\end{corollary}

We remark that theorems \ref{The general conservative system} and \ref{The
Noether Theorem} do not apply to generalized symmetries\cite{Sarlet,Kalotas}.

\section{Lie point symmetries of three dimensional autonomous Newtonian
systems}

\label{Lie point symmetries of three dimensional autonomous Newtonian
systems}

In this section we determine the forces $F^{\mu }=$ $F^{\mu }\left( x^{\nu
}\right) $ for which the equations of motion (\ref{PP.01}) admit Lie point
symmetries (in addition to the trivial $\partial _{t}).$ To do that we need
the special projective algebra of the Euclidian 3d metric

\begin{equation}
ds_{E}^{2}=dx^{2}+dy^{2}+dz^{2}.
\end{equation}

This algebra consists of 15 vectors\footnote{%
These vectors are not all linearly independent i.e. the HV and the rotations
are linear combinations of the ACs} as follows: Six KVs $\partial _{\mu
}~,~x_{\nu }\partial _{\mu }-x_{\mu }\partial _{\nu }~$ one HV$~R\partial
_{R},~$nine ACs~$x_{\mu }\partial _{\mu }~,~x_{\nu }\partial _{\mu }$ and
three SPCs $x_{\mu }^{2}\partial _{\mu }+x_{\mu }x_{\nu }\partial _{\nu
}+x_{\mu }x_{\sigma }\partial _{\sigma },~$where\footnote{%
If $x_{\mu }=x,~$then~$\left\{ x_{\nu }=y~,~x_{\sigma }=z\right\} ~$or~$%
\left\{ x_{\nu }=z~,~x_{\sigma }=y\right\} $}$~\mu \neq \nu \neq \sigma $ , $%
r_{\left( \mu \nu \right) }^{2}=x_{\mu }^{2}+x_{\nu }^{2},~\theta _{\left(
\mu \nu \right) }=\arctan \left( \frac{x_{\nu }}{x_{\mu }}\right) $ and $%
R,\theta ,\phi $ are spherical coordinates$.$

In the computation of Lie symmetries we consider only the linearly
independent vectors of the special projective group\footnote{%
We do not consider their linear combinations because the resulting Lie
symmetries are too many; on the other hand they can be computed in the
standard way.}.

\subsection{Lie point symmetries for non conservative forces}

In the following tables we list the Lie point symmetries and the functional
dependence of the components of the force for each Case of Theorem \ref{The
general conservative system} (for details concerning Cases A1, A2, A3 see
\cite{Tsam10}).

\begin{center}
Table 1. Case A1: Lie point symmetries generated by the affine
algebra\bigskip

$%
\begin{tabular}{|l|l|l|l|}
\hline
\textbf{Lie symmetry} & $\mathbf{F}_{\mu }\left( x_{\mu },x_{\nu },x_{\sigma
}\right) $ & $\mathbf{F}_{\nu }\left( x_{\mu },x_{\nu },x_{\sigma }\right) $
& $\mathbf{F}_{\sigma }\left( x_{\mu },x_{\nu },x_{\sigma }\right) $ \\
\hline
$\frac{d}{2}t\partial _{t}+\partial _{\mu }$ & $e^{-dx_{\mu }}f\left( x_{\nu
},x_{\sigma }\right) $ & $e^{-dx_{\mu }}g\left( x_{\nu },x_{\sigma }\right) $
& $e^{-dx_{\mu }}h\left( x_{\nu },x_{\sigma }\right) $ \\ \hline
$\frac{d}{2}t\partial _{t}+\partial _{\theta _{\left( \mu \nu \right) }}$ & $%
e^{-d\theta _{\left( \mu \nu \right) }}f\left( r_{\left( \mu \nu \right)
},x_{\sigma }\right) $ & $e^{-d\theta _{\left( \mu \nu \right) }}g\left(
r_{\left( \mu \nu \right) },x_{\sigma }\right) $ & $e^{-d\theta _{\left( \mu
\nu \right) }}h\left( r_{\left( \mu \nu \right) },x_{\sigma }\right) $ \\
\hline
$\frac{d}{2}t\partial _{t}+R\partial _{R}$ & $x_{\mu }^{1-d}f\left( \frac{%
x_{\nu }}{x_{\mu }},\frac{x_{\sigma }}{x_{\mu }}\right) $ & $x_{\mu
}^{1-d}g\left( \frac{x_{\nu }}{x_{\mu }},\frac{x_{\sigma }}{x_{\mu }}\right)
$ & $x_{\mu }^{1-d}h\left( \frac{x_{\nu }}{x_{\mu }},\frac{x_{\sigma }}{%
x_{\mu }}\right) $ \\ \hline
$\frac{d}{2}t\partial _{t}+x_{\mu }\partial _{\mu }$ & $x_{\mu
}^{1-d}f\left( x_{\nu },x_{\sigma }\right) $ & $x_{\mu }^{1-d}g\left( x_{\nu
},x_{\sigma }\right) $ & $x_{\mu }^{1-d}h\left( x_{\nu },x_{\sigma }\right) $
\\ \hline
$\frac{d}{2}t\partial _{t}+x_{\nu }\partial _{\mu }$ & $e^{-d\frac{x_{\mu }}{%
x_{\nu }}}\left[ \frac{x_{\mu }}{x_{\nu }}g\left( x_{\nu },x_{\sigma
}\right) +f\left( x_{\nu },x_{\sigma }\right) \right] $ & $e^{-d\frac{x_{\mu
}}{x_{\nu }}}g\left( x_{\nu },x_{\sigma }\right) $ & $e^{-d\frac{x_{\mu }}{%
x_{\nu }}}h\left( x_{\nu },x_{\sigma }\right) $ \\ \hline
\end{tabular}%
$ \\[0pt]

\bigskip

Table 2. Case A2: Lie point symmetries are generated by the gradient
homothetic algebra \bigskip

$%
\begin{tabular}{|l|l|l|l|}
\hline
\textbf{Lie symmetry} & $\mathbf{F}_{\mu }\left( x_{\mu },x_{\nu },x_{\sigma
}\right) $ & $\mathbf{F}_{\nu }\left( x_{\mu },x_{\nu },x_{\sigma }\right) $
& $\mathbf{F}_{\sigma }\left( x_{\mu },x_{\nu },x_{\sigma }\right) $ \\
\hline
$t\partial _{\mu }$ & $f\left( x_{\nu },x_{\sigma }\right) $ & $g\left(
x_{\nu },x_{\sigma }\right) $ & $h\left( x_{\nu },x_{\sigma }\right) $ \\
\hline
$t^{2}\partial _{t}+tR\partial _{R}$ & $\frac{1}{x_{\mu }^{3}}f\left( \frac{%
x_{\nu }}{x_{\mu }},\frac{x_{\sigma }}{x_{\mu }}\right) $ & $\frac{1}{x_{\mu
}^{3}}g\left( \frac{x_{\nu }}{x_{\mu }},\frac{x_{\sigma }}{x_{\mu }}\right) $
& $\frac{1}{x_{\mu }^{3}}h\left( \frac{x_{\nu }}{x_{\mu }},\frac{x_{\sigma }%
}{x_{\mu }}\right) $ \\ \hline
$e^{\pm t\sqrt{m}}\partial _{\mu }$ & $-mx_{\mu }+f\left( x_{\nu },x_{\sigma
}\right) $ & $g\left( x_{\nu },x_{\sigma }\right) $ & $h\left( x_{\nu
},x_{\sigma }\right) $ \\ \hline
$\frac{1}{\sqrt{m}}e^{\pm t\sqrt{m}}\partial _{t}\pm e^{\pm t\sqrt{m}%
}R\partial _{R}$ & $-\frac{m}{4}x_{\mu }+\frac{1}{x_{\mu }^{3}}f\left( \frac{%
x_{\nu }}{x_{\mu }},\frac{x_{\sigma }}{x_{\mu }}\right) $ & $-\frac{m}{4}%
x_{\nu }+\frac{1}{x_{\mu }^{3}}g\left( \frac{x_{\nu }}{x_{\mu }},\frac{%
x_{\sigma }}{x_{\mu }}\right) $ & $-\frac{m}{4}x_{\sigma }+\frac{1}{x_{\mu
}^{3}}h\left( \frac{x_{\nu }}{x_{\mu }},\frac{x_{\sigma }}{x_{\mu }}\right) $
\\ \hline
\end{tabular}%
$
\end{center}

For the remaining Case A3 we have that the force $F^{\mu }$ is the isotropic
oscillator, that is, ~$F^{\mu }=\left( \omega x^{\mu }+c^{\mu }\right)
\partial _{\mu }$ where $\omega ,~c^{\mu }$ are constants.

In order to demonstrate the use of the above tables let us require the
equations of motion of a Newtonian dynamical system which is invariant under
the $sl(2,R)$ algebra. We know \cite{Leach1991} that $sl(2,R)$ is generated
by the following Lie symmetries%
\[
\partial _{t}~,~2t\partial _{t}+R\partial _{R}~,t^{2}\partial
_{t}+tR\partial _{R}.
\]%
From \textbf{\ }table 1 line 3 for $d=4$\ and from Table 2 line 2 we have
that the force must be of the form\textbf{\ }\cite{Tsam10}%
\begin{equation}
F=\left( \frac{1}{x_{\mu }^{3}}f\left( \frac{x_{\nu }}{x_{\mu }},\frac{%
x_{\sigma }}{x_{\mu }}\right) ,\frac{1}{x_{\mu }^{3}}g\left( \frac{x_{\nu }}{%
x_{\mu }},\frac{x_{\sigma }}{x_{\mu }}\right) ,\frac{1}{x_{\mu }^{3}}h\left(
\frac{x_{\nu }}{x_{\mu }},\frac{x_{\sigma }}{x_{\mu }}\right) \right)
\end{equation}%
hence the equations of motion of this system in Cartesian coordinates are:%
\begin{equation}
(\ddot{x},\ddot{y},\ddot{z})=\left( \frac{1}{x_{\mu }^{3}}f\left( \frac{%
x_{\nu }}{x_{\mu }},\frac{x_{\sigma }}{x_{\mu }}\right) ,\frac{1}{x_{\mu
}^{3}}g\left( \frac{x_{\nu }}{x_{\mu }},\frac{x_{\sigma }}{x_{\mu }}\right) ,%
\frac{1}{x_{\mu }^{3}}h\left( \frac{x_{\nu }}{x_{\mu }},\frac{x_{\sigma }}{%
x_{\mu }}\right) \right) .
\end{equation}%
Immediately we recognize that this dynamical system is the well known and
important generalized Kepler Ermakov system (see \cite{Leach1991}). A\
different representation of $sl(2,R)$ consists of the vectors\cite{Tsam10}
\[
\partial _{t}~,~\frac{1}{\sqrt{m}}e^{\pm t\sqrt{m}}\partial _{t}\pm e^{\pm t%
\sqrt{m}}R\partial _{R}
\]%
For this representation from table 2 line 4 we have%
\begin{equation}
F^{\prime }=-\frac{m}{4}\left( x_{\mu },x_{\nu },x_{\sigma }\right) +\left(
\frac{1}{x_{\mu }^{3}}f\left( \frac{x_{\nu }}{x_{\mu }},\frac{x_{\sigma }}{%
x_{\mu }}\right) ,\frac{1}{x_{\mu }^{3}}g\left( \frac{x_{\nu }}{x_{\mu }},%
\frac{x_{\sigma }}{x_{\mu }}\right) ,\frac{1}{x_{\mu }^{3}}h\left( \frac{%
x_{\nu }}{x_{\mu }},\frac{x_{\sigma }}{x_{\mu }}\right) \right)
\end{equation}%
which leads again to the autonomous Kepler Ermakov system.

\subsection{Lie point symmetries for conservative forces}

In this section we assume that the force is given by the potential ~$%
V=V\left( x^{\mu }\right) $ and repeat the calculations. Again we ignore the
linear combinations of Lie symmetries for each case. We state the results in
Tables 3 and 4.

\begin{center}
Table 3. Case A1: Lie point symmetries generated by the affine algebra
(conservative force)\bigskip

$%
\begin{tabular}{|l|l|l|l|}
\hline
{\small {\textbf{Lie /}V(x,y,z) }} & $\mathbf{d=0}$ & $\mathbf{d=2}$ & $%
\mathbf{d\neq 0,2}$ \\ \hline
$\frac{d}{2}t\partial _{t}+\partial _{\mu }$ & $c_{1}x_{\mu }+f\left( x_{\nu
},x_{\sigma }\right) $ & $e^{-2x_{\mu }}f\left( x_{\nu },x_{\sigma }\right) $
& $e^{-dx_{\mu }}f\left( x_{\nu },x_{\sigma }\right) $ \\ \hline
$\frac{d}{2}t\partial _{t}+\partial _{\theta _{\left( \mu \nu \right) }}$ & $%
\,c_{1}\theta _{\left( \mu \nu \right) }+f\left( r_{\left( \mu \nu \right)
},x_{\sigma }\right) $ & $e^{-2\theta _{\left( \mu \nu \right) }}f\left(
r_{\left( \mu \nu \right) },x_{\sigma }\right) $ & $e^{-d\theta _{\left( \mu
\nu \right) }}f\left( r_{\left( \mu \nu \right) },x_{\sigma }\right) $ \\
\hline
$\frac{d}{2}t\partial _{t}+R\partial _{R}$ & $x^{2}f\left( \frac{x_{\nu }}{%
x_{\mu }},\frac{x_{\sigma }}{x_{\mu }}\right) $ & $c_{1}\ln \left( x_{\mu
}\right) +f\left( \frac{x_{\nu }}{x_{\mu }},\frac{x_{\sigma }}{x_{\mu }}%
\right) $ & $x^{2-d}f\left( \frac{x_{\nu }}{x_{\mu }},\frac{x_{\sigma }}{%
x_{\mu }}\right) $ \\ \hline
$\frac{d}{2}t\partial _{t}+x_{\mu }\partial _{\mu }$ & $c_{1}x_{\mu
}^{2}+f\left( x_{\nu },x_{\sigma }\right) $ & $\nexists $ & $\nexists $ \\
\hline
$\frac{d}{2}t\partial _{t}+x_{\nu }\partial _{\mu }$ & $c_{1}x_{\mu
}+c_{2}\left( x_{\mu }^{2}+x_{\nu }^{2}\right) +f\left( x_{\sigma }\right) $
& $\nexists $ & $\nexists $ \\ \hline
\end{tabular}%
$

\bigskip

Table 4. Case A2: Lie point symmetries generated by the gradient homothetic
algebra (conservative force) \bigskip

$%
\begin{tabular}{|l|l||l|l|}
\hline
\textbf{Lie} & $\mathbf{V}\left( x,y,z\right) $ & \textbf{Lie } & $\mathbf{V}%
\left( x,y,z\right) $ \\ \hline
$t\partial _{\mu }$ & $c_{1}x_{\mu }+f\left( x_{\nu },x_{\sigma }\right) $ &
$e^{\pm t\sqrt{m}}\partial _{\mu }$ & $-\frac{m}{2}x_{\mu }^{2}+c_{1}x_{\mu
}+f\left( x_{\nu },x_{\sigma }\right) $ \\ \hline
$t^{2}\partial _{t}+tR\partial _{R}$ & $\frac{1}{x_{\mu }^{2}}f\left( \frac{%
x_{\nu }}{x_{\mu }},\frac{x_{\sigma }}{x_{\mu }}\right) $ & $\frac{1}{\sqrt{m%
}}e^{\pm t\sqrt{m}}\partial _{t}+e^{\pm t\sqrt{m}}R\partial _{R}$ & $-\frac{m%
}{8}\left( x_{\mu }^{2}+x_{\nu }^{2}+x_{\sigma }^{2}\right) +\frac{1}{x_{\mu
}^{2}}f\left( \frac{x_{\nu }}{x_{\mu }},\frac{x_{\sigma }}{x_{\mu }}\right) $
\\ \hline
\end{tabular}%
$
\end{center}

Case B1/B2. In this case the potential is of the form $V\left( x,y,z\right) =%
\frac{\omega ^{2}}{2}\left( x^{2}+y^{2}+z^{2}\right) +p\left( x+y+z\right) ~$%
where $\omega ,p$ are constants.

From Tables 3 and 4 we infer{\LARGE \ }that the isotropic oscillator admits
24 Lie point symmetries generating the $Sl\left( 5,R\right) $, as many as
the free particle \cite{Prince}.

\section{Three dimensional autonomous Newtonian systems\newline
which admit Noether point symmetries}

\label{Noether point symmetries1}

In this section using theorem \ref{The Noether Theorem} we determine all
autonomous Newtonian Hamiltonian systems with Lagrangian%
\begin{equation}
L=\frac{1}{2}\left( \dot{x}^{2}+\dot{y}^{2}+\dot{z}^{2}\right) -V\left(
x,y,z\right)  \label{NPC.09}
\end{equation}%
which admit a non-trivial Noether point symmetry.This problem has been
considered previously in \cite{Damianou Sophocleous 1999,Damianou2004},
however as we shall show the results in these works are not complete. We
note that the Lie symmetries of a conservative system are not necessarily
Noether symmetries. The inverse is of course true.

Before we continue we note that the homothetic algebra of the Euclidian 3d
space $E^{3}$ has dimension seven and consists of three gradient KVs $%
\partial _{\mu }~~$with gradient function $x_{\mu }$, three non-gradient KVs
$x_{\nu }\partial _{\mu }-x_{\mu }\partial _{\nu }$ generating the
rotational algebra $so\left( 3\right) ,~$and a gradient HV $H^{i}=R\partial
_{R}~$\ with gradient function $H=\frac{1}{2}R^{2}$ , where $R^{2}=x^{\mu
}x_{\mu }.$ According to theorem \ref{The Noether Theorem} we have to
consider the following cases.

\subsection{Case I: Noether symmetries generated from the homothetic algebra}

The Noether point symmetries generated from the homothetic algebra i.e. the
non-gradient $so(3)$ elements included, are shown in Table 5.

\begin{center}
Table 5: Noether point symmetries generated by the homothetic algebra\bigskip

\begin{tabular}{|l|l|}
\hline
\textbf{Noether Symmetry} & $\mathbf{V(x,y,z)}$ \\ \hline
$\partial _{\mu }$ & $-px_{\mu }+f\left( x^{\nu },x^{\sigma }\right) $ \\
\hline
$x_{\nu }\partial _{\mu }-x_{\mu }\partial _{\nu }$ & $-p\theta _{\left( \mu
\nu \right) }+f\left( r_{\left( \mu \nu \right) },x^{\sigma }\right) ~$ \\
\hline
$2t\partial _{t}+R\partial _{R}$ & $\frac{1}{R^{2}}f\left( \theta ,\phi
\right) ~$or $\frac{1}{x_{\mu }^{2}}f\left( \frac{x_{\nu }}{x_{\mu }},\frac{%
x_{\sigma }}{x_{\mu }}\right) $ \\ \hline
\end{tabular}
\end{center}

The corresponding Noether integrals are computed easily from relation (\ref%
{NPC.07}) of Corollary \ref{The Noether Integrals}. In appendix \ref%
{appendix1} in Table 8 and Table 9 we give a complete list of the potentials
resulting form the linear combinations of the elements of the homothetic
algebra.

\subsection{Case II: Noether point symmetries generated from the gradient
homothetic algebra}

The Noether symmetries generated from the gradient homothetic algebra are
listed in Table 6.

\begin{center}
Table 6: Noether point symmetries generated by the gradient homothetic
algebra\bigskip

\begin{tabular}{|l|l|}
\hline
\textbf{Noether Symmetry} & $\mathbf{V(x,y,z)~~/~~T}_{,tt}\mathbf{=mT}$ \\
\hline
$T\left( t\right) \partial _{\mu }$ & $-\frac{m}{2}x_{\mu }^{2}-px_{\mu
}+f\left( x_{\nu },x_{\sigma }\right) $ \\ \hline
$\left( 2\int T\left( t\right) dt\right) \partial _{t}+T\left( t\right)
R\partial _{R}$ & $-\frac{m}{8}R^{2}+\frac{1}{R^{2}}f\left( \theta ,\phi
\right) $ or~$-\frac{m}{8}R^{2}+\frac{1}{x_{\mu }^{2}}f\left( \frac{x_{\nu }%
}{x_{\mu }},\frac{x_{\sigma }}{x_{\mu }}\right) $ \\ \hline
\end{tabular}
\end{center}

As before the Noether integrals corresponding to these Noether point
symmetries are computed from relation (\ref{NPC.08}) of Corollary \ref{The
Noether Integrals}. In appendix \ref{appendix1} in Table 10 we give the
potential functions which result from the linear combinations of the
elements of the gradient homothetic algebra. From the Tables we infer that
the isotropic linear forced oscillator admits 12 Noether point symmetries,
as many as the free particle.

As it has been remarked above, the determination of the Noether point
symmetries admitted by an autonomous Newtonian Hamiltonian system has been
considered previously in \cite{Damianou2004}. Our results extend the results
of \cite{Damianou2004} and coincide with them if we set the constant $p=0.~$%
For example in page 12 case 1 and page 15 case 6 of \cite{Damianou2004} the
terms $-\frac{p}{a}x_{\mu }$ and $p\arctan \left( l\left( \theta ,\phi
\right) \right) $ are missing respectively. Furthermore the potential given
in page/line 12/1, 13/2, 13/3 of \cite{Damianou2004} admits Noether
symmetries only when \ $\lambda =0~$and $b_{1,2}\left( t\right) =const.$
This is due to the fact that the vectors given in \cite{Damianou2004} are
KVs and in order to have $b_{,t}\neq 0$ they must be given by Case II\ of
theorem 2 above, that is, the KVs must be gradient. However the KVs used are
linear combinations of translations and rotations which are non-gradient.

We remark that from the above results we are also able to give, without any
further calculations, the Lie and the Noether point symmetries of a
dynamical system 'moving' in a three dimensional flat space whose metric has
Lorenzian signature simply by taking one of the coordinates to be complex,
for example by setting $x^{1}=ix^{1}.$

\section{Motion on the two dimensional sphere}

\label{Motion on the two dimensional sphere}

\textbf{\ }\label{Noether in a space of constant curvature}

A first application of the results of section \ref{Noether point symmetries1}
is the determination of Lie and Noether point symmetries admitted by the
equations of motion of a Newtonian particle moving in a two dimensional
space of constant non-vanishing curvature.

  Before we continue it is useful to recall some facts
concerning spaces of constant curvature. Consider a $n+1$ dimensional  flat
space with fundamental form%
\[
ds^{2}=\sum_{a}c_{a}(dz^{a})^{2}\quad a=1,2...,n+1
\]%
where $c_{a}$ are real constants. The hypersurfaces defined by
\[
\sum_{a}c_{a}(dz^{a})^{2}=eR_{0}^{2}
\]%
where $R_{0}$\ is an arbitrary constant and $e=\pm 1$  are called
fundamental hyperquadrics of the space. When all
coefficients $c_{a}$ are positive the space is Euclidian  and  $e=+1$. In this case
there is one family of hyperquadrics which is the hyperspheres. In all
other cases (excluding the case when all $c_{a}$ 's are negative) there are
two families of hyperquadrics corresponding to the values  $%
e=+1$ \ and $e=-1.$ It has been shown that in all cases the hyperquadrics are spaces of constant
curvature (see \cite{Eisenhart} p202).

One way to work is to consider in the above results $R=$constant. However,
in order to demonstrate the application of theorem \ref{The Noether Theorem}
in practice, we choose to work in the standard way. We use spherical
coordinates which are natural in the case of spaces of constant curvature.

We consider an autonomous dynamical system moving in the two dimensional
sphere (Euclidian $\left( \varepsilon =1\right) ~$or Hyperbolic~$\left(
\varepsilon =-1\right) $) with Lagrangian \cite{Carinena}%
\begin{equation}
L\left( \phi ,\theta ,\dot{\phi},\dot{\theta}\right) =\frac{1}{2}\left( \dot{%
\phi}^{2}+\mathrm{Sinn}^{2}\phi ~\dot{\theta}^{2}\right) -V\left( \theta
,\phi \right)  \label{MCC.01}
\end{equation}%
where
\[
\mathrm{Sinn}\phi =\left\{
\begin{array}{cc}
\mathrm{\sin }\phi & \varepsilon =1 \\
\mathrm{\sinh }\phi & \varepsilon =-1%
\end{array}%
\right. ~,~\mathrm{Cosn}\phi =\left\{
\begin{array}{cc}
\cos \phi & \varepsilon =1 \\
\mathrm{\cosh }\phi & \varepsilon =-1.%
\end{array}%
\right. ~
\]

The equations of motion are%
\begin{eqnarray}
\ddot{\phi}-\mathrm{Sinn}\phi ~\mathrm{Cosn}\phi \mathrm{~}\dot{\theta}%
^{2}+V_{,\phi } &=&0  \label{MCC.02} \\
\ddot{\theta}+2\frac{\mathrm{Cosn}\phi }{\mathrm{Sinn}\phi }~\dot{\theta}%
\dot{\phi}+\frac{1}{\mathrm{Sinn}^{2}\phi }V_{,\theta } &=&0.  \label{MCC.03}
\end{eqnarray}

We note that the Lagrangian (\ref{MCC.01}) is of the form (\ref{NPC.02})
with the metric $g_{\mu \nu }~$to be the metric of a space of constant
curvature. Therefore theorem \ref{The Noether Theorem} applies and we use it
to find the potentials $V\left( \theta ,\phi \right) $ for which additional
Noether point symmetries, hence Noether integrals are admitted.

The homothetic algebra of a metric of spaces of constant curvature consists
only of non-gradient KVs (hence $\psi =0)$ as follows

(a) $\varepsilon =1~$ (Euclidian case)

\begin{equation}
CK_{e}^{1}=\sin \theta \partial _{\phi }+\cos \theta \cot \phi \partial
_{\theta },~CK_{e}^{2}=\cos \theta \partial _{\phi }-\sin \theta \cot \phi
\partial _{\theta },~CK_{e}^{3}=\partial _{\theta }  \label{MCC.03.1}
\end{equation}

(b) $\varepsilon =-1~$(Hyperbolic case)%
\begin{equation}
CK_{h}^{1}=\sin \theta \partial _{\phi }+\cos \theta ~\mathrm{\coth }\phi
\partial _{\theta },~CK^{2}=\cos \theta \partial _{\phi }-\sin \theta ~%
\mathrm{\coth }\phi \partial _{\theta },~CK^{3}=\partial _{\theta }.
\label{MCC.03.2}
\end{equation}

Because we have only non-gradient KVs, according to theorem \ref{The Noether
Theorem} only Case I survives. Therefore the Noether vectors and the Noether
function are%
\begin{equation}
\mathbf{X}=CK_{e,h}^{i}\partial _{i},~~f=pt
\end{equation}%
provided the potential satisfies the condition%
\begin{equation}
\mathcal{L}_{CK}V+p=0.  \label{MCC.04}
\end{equation}%
The first integrals given by (\ref{NPC.07}) are
\begin{equation}
\phi _{II}=-g_{ij}^{i}CK_{e,h}^{i}\dot{x}^{j}+pt  \label{MCC.05}
\end{equation}%
and are time dependent if $p\neq 0$.

\subsection{Noether Symmetries}

We consider two cases, the case $V(\theta ,\phi )=$constant which concerns
the geodesics of the space, and the case $V(\theta ,\phi )\neq $constant.

For the case of geodesics it has been shown \cite{Tsamparlis2010} that the
Noether point symmetries are the three elements of $so(3)$ with
corresponding Noether integrals%
\begin{eqnarray}
I_{CK_{e,h}^{1}} &=&\dot{\phi}\sin \theta +\dot{\theta}\cos \theta ~\mathrm{%
Sinn}\phi ~\mathrm{Cosn}\phi \\
I_{CK_{e,h}^{2}} &=&\dot{\phi}\cos \theta -\dot{\theta}\sin \theta ~\mathrm{%
Sinn}\phi ~\mathrm{Cosn}\phi \\
I_{CK_{e,h}^{3}} &=&\dot{\theta}~\mathrm{Sinn}^{2}\phi .
\end{eqnarray}

These integrals are in involution with the Hamiltonian hence the system is
Liouville integrable.

In the case $V(\theta ,\phi )\neq $constant we find the results of Table 7

\begin{center}
Table 7: Noether symmetries/Integrals and potentials for the Lagrangian (\ref%
{MCC.01})~

$%
\begin{tabular}{|l|l|l|}
\hline
\textbf{Noether~Symmetry} & $\mathbf{V}\left( \theta ,\phi \right) $ &
\textbf{Noether~Integral} \\ \hline
$CK_{e,h}^{1}$ & $F\left( \cos \theta ~\mathrm{Sinn}\phi \right) $ & $%
I_{CK_{e,h}^{1}}$ \\ \hline
$CK_{e,h}^{2}$ & $F\left( \sin \theta ~\mathrm{Sinn}\phi \right) $ & $%
I_{CK_{e,h}^{2}}$ \\ \hline
$CK_{e,h}^{3}$ & $F\left( \phi \right) $ & $I_{CK_{e,h}^{3}}$ \\ \hline
$aCK_{e,h}^{1}+bCK_{e,h}^{2}$ & $F\left( \frac{1+\tan ^{2}\theta }{\mathrm{%
Sinn}^{2}\phi ~\left( a-b\tan \theta \right) ^{2}}\right) $ & $%
aI_{CK_{e,h}^{1}}+bI_{CK_{e,h}^{2}}$ \\ \hline
$aCK_{e,h}^{1}+bCK_{e,h}^{3}$ & $F\left( a\cos \theta \mathrm{Sinn}\phi
-\varepsilon ~b~\mathrm{Cosn}\phi \right) $ & $%
aI_{CK_{e,h}^{1}}+bI_{CK_{e,h}^{3}}$ \\ \hline
$aCK_{e,h}^{2}+bCK_{e,h}^{3}$ & $F\left( a\sin \theta \mathrm{Sinn}\phi
-\varepsilon ~b~\mathrm{Cosn}\phi \right) $ & $%
aI_{CK_{e,h}^{2}}+bI_{CK_{e,h}^{3}}$ \\ \hline
$aCK_{e,h}^{1}+bCK_{e,h}^{2}+cCK_{e,h}^{3}$ & $F\left( \left( a\cos \theta
-b\sin \theta \right) ~\mathrm{Sinn}\phi -\varepsilon ~c~\mathrm{Cosn}\phi
\right) $ & \thinspace $%
aI_{CK_{e,h}^{1}}+bI_{CK_{e,h}^{2}}+cI_{CK_{e,h}^{3}} $ \\ \hline
\end{tabular}%
$
\end{center}

The first integrals which correspond to each potential of Table 7 are in
involution with the Hamiltonian and independent. Hence the corresponding
systems are integrable. From Table 7 we infer the following result.

\begin{proposition}
\label{prop CCs}A dynamical system with Lagrangian (\ref{MCC.01})~has one,
two or four Noether point symmetries hence Noether integrals.
\end{proposition}

\begin{proof}
For the case of the free particle we have the maximum number of four Noether
symmetries (the rotation group $so(3)$ plus the $\partial _{t}$). In the
case the potential is not constant the Noether symmetries are produced by
the non-gradient KVs with Lie algebra%
\[
\left[ X_{A},X_{B}\right] =C_{AB}^{C}X_{C}
\]%
where $C_{12}^{3}=C_{31}^{2}=~C_{23}^{1}=1$ for $\varepsilon =1~$and $\bar{C}%
_{21}^{3}=\bar{C}_{23}^{1}=\bar{C}_{31}^{2}=1$ for $\varepsilon =-1.~$%
Because the Noether point symmetries form a Lie algebra and the Lie algebra
of the KVs is semisimple the system will admit either none, one or three
Noether symmetries generated from the KVs. The case of three is when $%
V\left( \theta ,\phi \right) =V_{0}$ that is the case of geodesics,
therefore the Noether point symmetries will be (including $\partial _{t}$)
either one, two or four.
\end{proof}

We note that the two important potentials of Celestial Mechanics,\ that is $%
V_{1}=-\frac{\mathrm{Cosn}\phi }{\mathrm{Sinn}\phi }~,~V_{2}=\frac{1}{2}%
\frac{\mathrm{Sinn}^{2}\phi }{\mathrm{Cosn}^{2}\phi }~\ $which according to
Bertrand 's Theorem \cite{Carinena,Kozlov,Vozmi} \ produce closed orbits on
the sphere are included in Table 7. Hence the dynamical systems they define
are Liouville integrable via Noether point symmetries $~CK_{e,h}^{3}$. \ The
potential $V_{1}$ corresponds to the Newtonian Kepler potential and $V_{2}$
is the analogue of the harmonic oscillator. We also note that our results
contain those of \cite{Carinena} if we consider the correspondence\footnote{%
We thank one of the referees for bringing this reference to our attention.} $%
S_{k}(r)\rightarrow \sin \phi ,C_{k}(r)\rightarrow \cos \phi ,$ $\theta
\rightarrow \phi ,$ $v_{r}\rightarrow \dot{\phi},$ $v_{\phi }\rightarrow
\dot{\theta}.$

We emphasize that the potentials listed in Table 7 concern dynamical systems
with Lagrangian (\ref{MCC.01}) which are integrable via Noether point
symmetries. It is possible that there exist integrable Newtonian dynamical
systems for potentials not included in these Tables, for example systems
which admit only dynamical symmetries \cite{Sarlet,Kalotas} with integrals
quadratic in momenta~\cite{Crampin,Lundmark}. However these systems are not
integrable via Noether point symmetries.

\section{Conclusion}

\label{Conclusion}

We have determined the three dimensional Newtonian dynamical systems which
admit Lie point symmetries and the three dimensional Hamiltonian Newtonian
dynamical systems which admit Noether point symmetries. These results
complete previous results \cite{Damianou Sophocleous 1999,Damianou2004}
concerning the Noether point symmetries of the three dimensional Newtonian
dynamical systems and extend our previous work on the two dimensional case
\cite{Tsam10,Sen}. We note that, due to the geometric derivation and the
tabular presentation, the results can be extended easily to higher
dimensional flat spaces, however at the cost of convenience because the
linear combinations of the symmetry vectors increase dramatically. In a
subsequent work, we shall apply the results obtained here to study the
integrability of the three dimensional Hamiltonian Kepler-Ermakov system
\cite{Leach1991} and generalize it in a Riemannian space.

\subsection*{Acknowledgement}

We would like to thank the anonymous referee for useful comments and
suggestions. This work has been partially supported from ELKE (grant 1112)
of the University of Athens.

\appendix{}

\section*{Appendix}

\label{appendix1}

Tables 8, 9 and 10 give the three dimensional potentials which admit Noether
point symmetries resulting from linear combinations of the elements of the
homothetic group.

\newpage 

\begin{center}
Table 8: Linear combinations of two vector fields for case I.\bigskip

\begin{tabular}{|l|l|}
\hline
\textbf{Noether Symmetry} & $\mathbf{V(x,y,z)}$ \\ \hline
$a\partial _{\mu }+b\partial _{\nu }$ & $-\frac{p}{a}x_{\mu }+f\left( x^{\nu
}-\frac{b}{a}x^{\mu },x^{\sigma }\right) $ \\ \hline
$a\partial _{\mu }+b\left( x_{\nu }\partial _{\mu }-x_{\mu }\partial _{\nu
}\right) $ & $-\frac{p}{\left\vert b\right\vert }\arctan \left( \frac{%
\left\vert b\right\vert x_{\mu }}{\left\vert \left( a+bx_{\nu }\right)
\right\vert }\right) +f\left( \frac{1}{2}r_{\left( \mu \nu \right) }+\frac{a%
}{b}x^{\nu },x^{\sigma }\right) $ \\ \hline
$a\partial _{\mu }+b\left( x_{\sigma }\partial _{\nu }-x_{\nu }\partial
_{\sigma }\right) $ & $-\frac{p}{\left\vert b\right\vert }\theta _{\left(
\nu \sigma \right) }+f\left( r_{\left( \nu \sigma \right) },x^{\mu }-\frac{a%
}{b}\theta _{\left( \nu \sigma \right) }\right) $ \\ \hline
$a\left( x_{\nu }\partial _{\mu }-x_{\mu }\partial _{\nu }\right) +$ & $%
\frac{p}{a}\arctan \left( \frac{ax_{\nu }+bx_{\sigma }}{x_{\mu }\sqrt{%
a^{2}+b^{2}}}\right) +$ \\
$~~~+b\left( x_{\sigma }\partial _{\mu }-x_{\mu }\partial _{\sigma }\right) $
& \thinspace $+\frac{1}{a}f\left( x_{\sigma }-\frac{a}{b}x_{\nu },x_{\nu
}^{2}\left( 1-\left( \frac{a}{b}\right) ^{2}+\frac{2b}{a}\frac{x_{\sigma }}{%
x_{\nu }}\right) +x_{\mu }^{2}\right) $ \\ \hline
$2bt\partial _{t}+a\partial _{\mu }+bR\partial _{R}$ & $-p\frac{x_{\mu
}\left( 2a+bx_{\mu }\right) }{2\left( a+bx_{\mu }^{2}\right) }+\frac{1}{%
\left( a+bx_{\mu }^{2}\right) }f\left( \frac{x_{\nu }}{a+bx_{\mu }},\frac{%
x_{\sigma }}{a+bx_{\mu }}\right) $ \\ \hline
$2bt\partial _{t}+a\theta _{\left( \mu \nu \right) }\partial _{\theta
_{\left( \mu \nu \right) }}+bR\partial _{R}$ & $\frac{1}{r_{\left( \mu \nu
\right) }^{2}}f\left( \theta _{\left( \mu \nu \right) }-\frac{a}{b}\ln
r_{\left( \mu \nu \right) },\frac{x_{\sigma }}{r_{\left( \mu \nu \right) }}%
\right) $ \\ \hline
\end{tabular}

\bigskip

Table 9: Linear combination of three vector fields for case I.\bigskip

\begin{tabular}{|l|l|}
\hline
\textbf{Noether Symmetry} & $\mathbf{V(x,y,z)}$ \\ \hline
$a\partial _{\mu }+b\partial _{\nu }+c\partial _{\sigma }$ & $-\frac{p}{a}%
x_{\mu }+f\left( x^{\nu }-\frac{b}{a}x^{\mu },x^{\sigma }-\frac{c}{a}x^{\mu
}\right) $ \\ \hline
$a\partial _{\mu }+b\partial _{\nu }+c\left( x_{\nu }\partial _{\mu }-x_{\mu
}\partial _{\nu }\right) $ & $\frac{p}{\left\vert c\right\vert }\arctan
\left( \frac{\left( b-cx_{\mu }\right) }{\left\vert \left( a+cx_{\nu
}\right) \right\vert }\right) $ \\
& ~$+f\left( \frac{c}{2}r_{\left( \mu \nu \right) }-bx_{\mu }+ax_{\nu
},x_{\sigma }\right) $ \\ \hline
$a\partial _{\mu }+b\partial _{\nu }+c\left( x_{\sigma }\partial _{\mu
}-x_{\mu }\partial _{\sigma }\right) $ & $-\frac{p}{\left\vert c\right\vert }%
\arctan \left( \frac{\left\vert c\right\vert x_{\mu }}{\left\vert
a+cx_{\sigma }\right\vert }\right) $ \\
& $~+f\left( x_{\nu }-\frac{1}{\left\vert c\right\vert }\arctan \left( \frac{%
\left\vert c\right\vert x_{\mu }}{\left\vert a+cx_{\sigma }\right\vert }%
\right) ,\frac{1}{2}r_{\left( \mu \sigma \right) }-\frac{a}{c}x_{\sigma
}\right) $ \\ \hline
$a\partial _{\mu }+b\left( x_{\nu }\partial _{\mu }-x_{\mu }\partial _{\nu
}\right) +$ & $\frac{p}{\sqrt{b^{2}+c^{2}}}\arctan \left( \frac{\left(
ab+b^{2}x_{\nu }+bcx_{\sigma }\right) }{\left\vert bx_{\mu }\right\vert
\sqrt{b^{2}+c^{2}}}\right) +$ \\
$~~+c\left( x_{\sigma }\partial _{\mu }-x_{\mu }\partial _{\sigma }\right) $
& $+f\left( x_{\mu }^{2}+x_{\nu }^{2}\left( 1-\frac{c^{2}}{b^{2}}\right)
+\left( \frac{2a}{b}+\frac{2c}{b}x_{\sigma }\right) x_{\nu },x_{\sigma }-%
\frac{c}{b}x_{\nu }\right) $ \\ \hline
$so\left( 3\right) $ linear combination & $p\arctan \left( \lambda \left(
\theta ,\phi \right) \right) +$ \\
& $+~F\left( R,b\tan \theta \sin \phi +c\cos \phi -aM_{1}\right) $ \\ \hline
$2ct\partial _{t}+a\partial _{\mu }+b\theta _{\left( \nu \sigma \right)
}\partial _{\theta _{\left( \nu \sigma \right) }}+cR\partial _{R}$ & $\frac{1%
}{r_{\left( \nu \sigma \right) }^{2}}f\left( \theta _{\left( \nu \sigma
\right) }-\frac{b}{c}\ln r_{\left( \nu \sigma \right) },\frac{a+cx_{\mu }}{%
cr_{\left( \nu \sigma \right) }}\right) $ \\ \hline
$2lt\partial _{t}+\left( a\partial _{\mu }+b\partial _{\nu }+c\partial
_{\sigma }+lR\partial _{R}\right) $ & $-\frac{px\left( 2a+cx_{\mu }\right) }{%
2\left( a+cx_{\mu }\right) ^{2}}+\frac{1}{\left( a+lx_{\mu }\right) ^{2}}%
f\left( \frac{b+lx_{\nu }}{l\left( a+lx_{\mu }\right) },\frac{c+lx_{\sigma }%
}{l\left( a+lx_{\mu }\right) }\right) $ \\ \hline
\end{tabular}%
\\[0pt]
\end{center}

where in Table 10
\begin{eqnarray*}
\lambda \left( \phi ,\theta \right) &=&\left( \left( a^{2}+b^{2}\right) \cos
\phi -bc\tan \theta \sin \phi +cM_{1}\right) \times \\
&&\times \left\{ M_{2}\left[ -b^{2}M_{1}^{2}-2b\tan \theta \sin \phi
M_{1}-a^{2}\sin ^{2}\phi \tan ^{2}\theta \right] \right\} ^{-\frac{1}{2}}
\end{eqnarray*}%
and~$M_{1}=\frac{1}{\cos \theta }\sqrt{\sin ^{2}\phi \left( 2\cos ^{2}\theta
-1\right) }~,M_{2}=\sqrt{a^{2}+b^{2}+c^{2}}.$

\begin{center}
Table 10: Linear combination of vector fields for case II.\bigskip

\begin{tabular}{|l|l|}
\hline
\textbf{Noether Symmetry} & $\mathbf{V(x,y,z)~/~T}_{,tt}\mathbf{=mT}$ \\
\hline
$T\left( t\right) \left( a\partial _{\mu }+b\partial _{\nu }+c\partial
_{\sigma }\right) $ & $-\frac{m}{2a}R^{2}+f\left( x^{\nu }-\frac{b}{a}x^{\mu
},x^{\sigma }-\frac{c}{a}x^{\mu }\right) $ \\ \hline
$\left( 2l\int T\left( t\right) dt\right) \partial _{t}+$ & $\frac{1}{\left(
a+lx_{\mu }\right) ^{2}}f\left( \frac{b+lx_{\nu }}{l\left( a+lx_{\mu
}\right) },\frac{c+lx_{\sigma }}{l\left( a+lx_{\mu }\right) }\right) +$ \\
$~~+T\left( t\right) \left( a\partial _{\mu }+b\partial _{\nu }+c\partial
_{\sigma }+lR\partial _{R}\right) $ & $~-\frac{m}{8}\left( R^{2}+\frac{2a}{l}%
x_{\mu }+\frac{2c}{l}x_{\nu }+\frac{2b}{l}x_{\sigma }\right) $ \\ \hline
\end{tabular}
\end{center}

\end{document}